\documentclass{aptpub}

\usepackage[colorlinks = true,linkcolor = blue,urlcolor  = blue,citecolor = blue,anchorcolor = blue]{hyperref}

\usepackage[english]{babel}
\usepackage[a4paper,top=2cm,bottom=2cm,left=3cm,right=3cm,marginparwidth=1.75cm]{geometry}
\usepackage{amsmath,amssymb}
\usepackage{graphicx}
\usepackage{booktabs}

\DeclareMathOperator{\E}{\mathbb{E}}

\authornames{J.~G.~Rasmussen {\it et al.}} 
\shorttitle{Modeling electrical distribution networks with inhomogeneous Galton-Watson trees} 

\begin{document}

\title{Modeling electrical distribution networks with inhomogeneous Galton-Watson trees}

\authorone[Aalborg University]{Jakob G. Rasmussen} 
\authortwo[Aalborg University]{Troels Pedersen}
\authorthree[Aalborg University]{Rasmus L. Olsen}
\addressone{Dept. Mathematical Sciences, Aalborg University, Thomas Manns Vej 23, 9220 Aalborg East, Denmark.} 
\emailone{jgr@math.aau.dk} 
\addresstwo{Dept. Electronic Systems, Aalborg University, Fredrik Bajers Vej 7A, 9220 Aalborg East, Denmark.} 
\emailtwo{troels@es.aau.dk} 
\addressthree{Dept. Electronic Systems, Aalborg University, Fredrik Bajers Vej 7A, 9220 Aalborg East, Denmark.} 
\emailthree{rlo@es.aau.dk}

\begin{abstract}
In this paper we consider inhomogeneous Galton-Watson trees, and derive various moments for such processes: the number of vertices, the number of leaves, and the height of the tree. Also we make a simple condition of finiteness. We use these processes to model a data set consisting of electrical distribution networks, where we make a flexible framework for formulating models through the mean and variance of the offspring distributions. Furthermore, we introduce two mixture distributions as offspring distributions to reflect the particular form of the data. For estimation we use maximum likelihood estimation.
\end{abstract}

\keywords{Branching process; maximum likelihood estimation; moments; network topology; offspring distributions}

\ams{60J80}{62P30}    

\section{Introduction}

Galton-Watson processes are classic stochastic processes used for modeling reproducing populations as branching processes due to their mathematical tractability. Many generalizations have been made to adapt these to various problems, such as Galton-Watson processes with immigration, age dependent Galton-Watson processes, or Galton-Watson processes in random environments. In the present paper, we consider inhomogeneous Galton-Watson processes (also known as Galton-Watson processes with varying environments or Galton-Watson process with generation dependence). We focus explicitly on the tree structures generated by these branching processes, and to reflect this we use the terminology Inhomogeneous Galton-Watson Tree (IGWT) throughout the paper.

Various papers study criteria for identifying whether an IGWT is super-critical (i.e.\ is infinite with positive probability) or not, e.g.\ \cite{bertacchi16}, or calculate or bound the probability of an infinite tree for the super-critical case, e.g.\ \cite{broman08}. 
Furthermore, limits for the super-critical case are studied in many papers, e.g.\ \cite{Foster76,gao15}. Less focus have been given to the sub-critical case of the IGWT, although papers such as \cite{fearn72} proves results that are equally important for all cases of an IGWT. 

In the present paper, we  focus on modeling the topology of electrical distribution networks, and deriving results for the IGWT relevant to such modeling. While  topologies of medium- and high-voltage transmission grids typically  have loops for redundancy, low-voltage distribution grids are predominantly tree-structures. For this reason we restrict our attention to low-voltage grids.
Furthermore, since distribution networks are small compared to other use-cases of Galton-Watson processes and their generalizations, we  focus on the sub-critical case. Finally, the topology of such networks are known (ignoring possible missing or faulty information), so the methods assume data where the tree structures are completely known. 

Other papers have studied different stochastic models for generating graphs representing the topology of electrical networks, for example: \cite{aksoy2019} considers a two-phase model with a Chung-Le random graph model for generating subgraphs and a star graph generation algorithm for connecting the subgraphs; \cite{giacomarra24} uses exponential graph models; \cite{shahraeini24} uses modifications of the Erd\"{o}s-R\'{e}ney model; \cite{ma17} considers random spanning trees and adds extra random edges afterwards; and \cite{schultz2014random}  employs a random tree growth model including spatial coordinates. However, all of these papers focus on models with loops, sometimes specifically aimed at medium- or high-voltage grids, while we are not aware of papers specifically focusing on making stochastic models for generating topologies of low-voltage grids. 

The data used in the present paper contain the topologies of 100 Danish electrical low voltage distribution networks from various places in the operational area of a single distribution system operator, ranging from rural farm areas to medium sized town areas, see \cite{dataset}. Three of the topologies are shown in Figure~\ref{fig.3grids}. It should be obvious from just these three examples that the networks vary quite a lot. For example, the number of vertices, height and width of the three trees differ massively. IGWTs similarly has an inherent large variation due to the branching structure, thus motivating the use of these as models for this particular data. 

\begin{figure}
    \centering
    \includegraphics[width=0.9\linewidth]{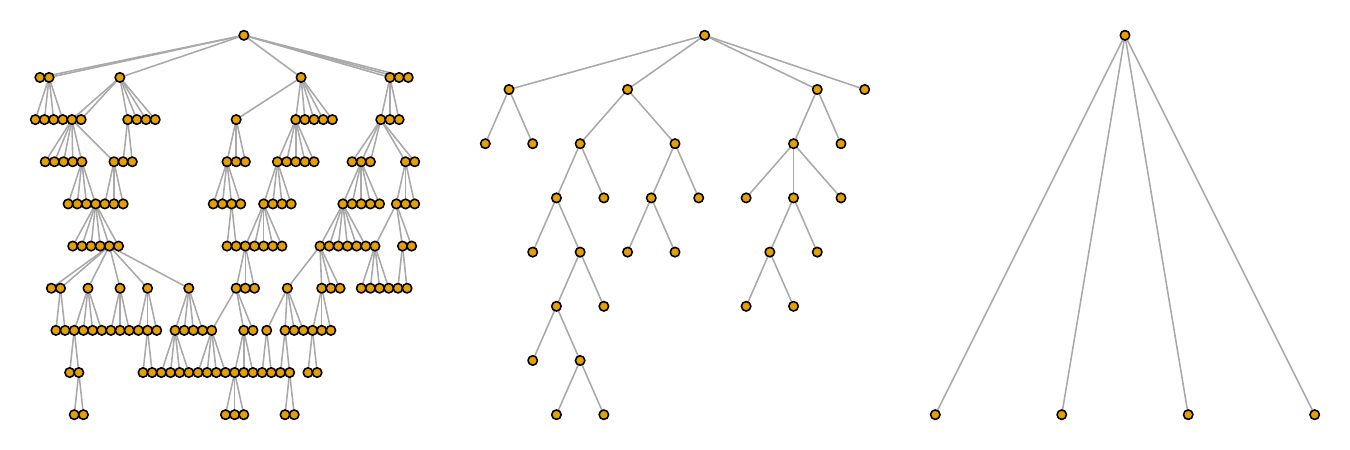}
    \caption{Three examples of grids from the data.}
    \label{fig.3grids}
\end{figure}

From an application point-of-view,  the model proposed in this paper can be used to easily simulate a large number of synthetic electrical distribution networks. Synthetic networks are needed for designing and testing smart grid solutions targeting the low voltage distribution grids. Currently such methods are proposed and tested in simulation studies on a few networks, which are either handcrafted test scenarios or real-life networks. As an example, the work \cite{olsen24} considers the estimation of voltages and currents on a single real-life distribution network. Other examples include voltage control systems or other smart grid solutions targeting the low voltage distribution grids, see e.g.\ \cite{Hassani04052025,KRISTENSEN2018122}, where optimal solutions are usually linked to the topology. Testing on a few hand-picked or hand-crafted networks can be problematic for several reasons. It is hard to say whether obtained numerical results would differ significantly for other networks. Furthermore the test approach can be criticized since a topology used may be picked or constructed to give an advantage or disadvantage in a performance analysis. Thus ideally, the test should be performed for many different networks, but handcrafting large numbers of grid topologies is not practical. Relying on real-life data from already existing grids would impede the ability to test performance in future grids.  Using a stochastic model such assessments may even be performed for different parameter settings of the model to assess how smart grid solutions will cope with potential future changes in the network topology.  


The paper is structured as follows: Section~\ref{sec.igwt} defines the IGWT and establishes important properties including various moments, probability generating functions and a simple sufficient condition for finiteness. Section~\ref{sec.models} introduces a modeling framework for constructing IGWT models for modeling data, and considers a few offspring distributions explicitly tailored to fit the distribution network data. The likelihood function is derived in Section~\ref{sec.mle}. The model is specified and fitted to the data in Section~\ref{sec.analysis}, and the fit of the model is investigated. Finally, concluding remarks are given in Section~\ref{sec.con}

\section{IGWTs and their properties}\label{sec.igwt}

\subsection{Definition}

The IGWT is defined recursively one generation at a time. Let $n=0,1,\ldots$ denote the generation number. To each generation associate a distribution, the so-called offspring distribution, with probability function $f_n$ defined on the non-negative integers. Assume that $f_n$ has mean $\mu_n$ and variance $\sigma_n^2$.
\begin{definition}\label{def.igwt}
The IGWT is given by the following recursive construction.
    \begin{enumerate}
        \item Generation 0 consists of one vertex.
        \item For $n=0,1,\ldots$, each vertex $i$ in generation $n$ has $X_{n,i}$ offspring drawn from $f_n$ independently of everything else.
        \item This ends when there no new offspring in a generation, and the output is then the obtained directed tree.
    \end{enumerate}
\end{definition}

Since the model is defined constructively in Definition~\ref{def.igwt}, it is easy to simulate it according to this definition (provided of course that a simulation algorithm is available for the distribution specified by $f_n$). Denote the total number of vertices in generation $n$ by $Z_n$ and the total number of vertices in the tree by $Z$. In general $Z$ may become infinite with positive probability, but we address this issue in Section~\ref{sec.finite}. 

\subsection{Moments of the number of vertices}

Denote the following moments of an IGWT by $m_n=\E(Z_n)$, $s_n^2=\var(Z_n)$, $c_{n,n+k}=\cov(Z_{n},Z_{n+k})$, $m=\E(Z)$, and $s^2=\var(Z)$. The following proposition then gives formulas for these, where a sum or a product over an empty index set should be interpreted as zero and one, respectively (i.e.\ $m_0=1$, $s_0^2=0$, $c_{0,k}=0$ for $k>0$ in the proposition). It should be noted that the expressions for the mean and variance of $Z_n$, i.e.\ \eqref{eq.meann} and \eqref{eq.varn}, are well-known and merely included for the sake of completeness (see e.g.\ Propositions~4 and 6 in \cite{fearn72}).

\begin{proposition}\label{prop.momn} The first and second moments of $Z_n$ for $n\geq 0$ are given by 
    \begin{align}
        m_n &= \prod_{i=0}^{n-1} \mu_i,\label{eq.meann}\\
        s_n^2 &= \sum_{i=0}^{n-1} \sigma_i^2 \left(\prod_{j=0}^{i-1} \mu_j\right)\left(\prod_{j=i+1}^{n-1} \mu_j^2\right),\label{eq.varn}\\
        c_{n,n+k} &= s_n^2 \frac{m_{n+k}}{m_n},\qquad \text{for }k>0.\label{eq.covnk}
    \end{align}
    Furthermore, the first and second moments of $Z$ are given by
    \begin{align}
		m &= \sum_{n=0}^\infty \prod_{i=0}^{n-1} \mu_i,\label{eq.meanall}\\
  		s^2 &= \sum_{n=1}^\infty \left(s_n^2 + 2\frac{s_n^2}{m_n}\sum_{k=1}^\infty m_{n+k}\right).\label{eq.varall}
    \end{align}

\end{proposition}
\begin{proof}

    To prove \eqref{eq.covnk}, consider a sum of a random number $N$ of i.i.d.\ random variables $Y_i$, where $N$ is independent of $Y_i$. Then by the law of total covariance
    \begin{align*}
    \cov\left(N,\sum_{i=1}^N Y_i\right) &= \E\left(\cov\left(\left.N,\sum_{i=1}^N Y_i\right|N\right)\right) + \cov\left(\E(N|N),\E\left(\left.\sum_{i=1}^N Y_i\right|N\right)\right)\\
    &= 0 + \cov(N,N\E(Y_i))\\
    &= \var(N)\E(Y_i).
    \end{align*}
    Thus 
    \[
    \cov(Z_n,Z_{n+k}) = \var(Z_n)\E(Z^*_{n+k}),
    \]
    where $Z_{n+k}^*$ denotes the number of vertices in generation $n+k$ coming from a single vertex in generation $n$. Therefore
    \[
    \cov(Z_n,Z_{n+k}) = s_n^2 \prod_{i=n}^{n+k-1}\mu_i = s_n^2 \frac{m_{n+k}}{m_n}.
    \]
    The mean and variance of $Z$ given by \eqref{eq.meanall} and \eqref{eq.varall} follow immediately from the usual formulas for means and variances of sums noting that
    all the terms for $n=0$ are 0 and thus disappears in the variance.
\end{proof}

\subsection{Moments of the number of leaves}

Rather than considering all vertices, we can also consider the leaves only. In the context of distribution networks, this corresponds to consumers. Letting $p_n=P(X_{n,i}=0)$, we get various moments of the number of leaves denoted by $\tilde m_n=\E(\tilde Z_n)$, $\tilde s_n^2 = \var(\tilde Z_n)$, $\tilde c_{n,n+k} = \cov(\tilde Z_n,\tilde Z_{n+k})$, $\tilde m = \E(\tilde Z)$, and $\tilde s^2 = \var(\tilde Z)$, where $\tilde Z_n$ and $\tilde Z$ denote the number of vertices without offspring in generation $n$ and in the tree, respectively.
\begin{proposition}
The first and second moments of $\tilde Z_n$ are given by 
    \begin{align}
        \tilde m_n &= m_n p_n,\label{eq.tmeann}\\
        \tilde s_n^2 &= s_n^2p_n^2 + m_n p_n (1-p_n),\label{eq.tvarn}\\
        \tilde c_{n,n+k} &= p_{n+k}p_nm_{n+k}\left(\frac{s_n^2}{m_n}-1\right) ,\qquad \text{for }k>0.\label{eq.tcovnk}
    \end{align}
    Furthermore, the first and second moments of $\tilde Z$ are given by
    \begin{align}
		\tilde m &= \sum_{n=0}^\infty m_np_n,\label{eq.tmeanall}\\
  		\tilde s^2 &= \sum_{n=0}^\infty p_n\left(m_n + \left(\frac{s_n^2}{m_n}-1\right)\left(m_np_n+2\sum_{k=1}^\infty p_{n+k}m_{n+k} \right)\right)\label{eq.tvarall}
    \end{align}
\end{proposition}

\begin{proof}
    Observe that $\tilde Z_n$ given $Z_n$ is binomially distributed with parameters $Z_n$ and $p_n$, and thus $\E(\tilde Z_n|Z_n) = Z_n p_n$ and $\var(\tilde Z_n|Z_n) = Z_np_n(1-p_n)$. Therefore the law of total expectation gives \eqref{eq.tmeann} by
    \[
    \tilde m_n = \E(\E(\tilde Z_n|Z_n)) = \E(Z_np_n) = m_n p_n,
    \]
    and the law of total variance gives \eqref{eq.tvarn} by
    \[
    \var(\tilde Z_n) = \var(\E(\tilde Z_n|Z_n)) + \E(\var(\tilde Z_n|Z_n) = \var(Z_np_n) + \E(Z_np_n(1-p_n)) = s_n^2p_n^2 + m_n p_n (1-p_n).
    \]
    To get \eqref{eq.tcovnk}, we start with the law of total covariance, which gives
    \begin{equation}\label{eq.lotc}
    \cov(\tilde Z_n,\tilde Z_{n+k}) = \E(\cov(\tilde Z_n, \tilde Z_{n+k}|Z_n)) + \cov(\E(\tilde Z_n|Z_n),\E(\tilde Z_{n+k}|Z_n)).
    \end{equation}
    We will start by considering the conditional covariance in the first term on the right hand side, which can be expanded as
    \begin{equation}\label{eq.lotcterm1}
    \cov(\tilde Z_n, \tilde Z_{n+k}|Z_n) = \E(\tilde Z_n\tilde Z_{n+k}|Z_n) - \E(\tilde Z_n|Z_n)\E(\tilde Z_{n+k}|Z_n).
    \end{equation}
    The last expectation on the right hand side of \eqref{eq.lotcterm1} is given by
    \begin{equation}\label{eq.lotcterm1-1}
    \E(\tilde Z_{n+k}|Z_n) = \E(\E(\tilde Z_{n+k}|Z_{n+k})|Z_n) = \E(p_{n+k}Z_{n+k}|Z_n) =  p_{n+k}\frac{m_{n+k}}{m_n} Z_n,
    \end{equation}
    where the last equality follows from the fact that we can regard $Z_{n+k}|Z_n$ as $Z_n$ independent IGWTs started at generation $n$. The first expectation on the right hand side of \eqref{eq.lotcterm1} is given by
    \[
    \E(\tilde Z_n\tilde Z_{n+k}|Z_n) = \E(\tilde Z_n \E(\tilde Z_{n+k}|\tilde Z_n)|Z_n).
    \]
    Conditioning on $\tilde Z_n$ and $Z_n$ means that we essentially know how many processes started at generation $n$ becomes leaves (i.e.\ $\tilde Z_n$) and how many continues with at least one offspring, and thus potentially can give leaves at generation $n+k$ (i.e.\ $Z_n-\tilde Z_n)$. If we denote the number of leaves at generation $n+k$ from a single process started at generation $n$ conditioned on at least one offspring by $\tilde Z_{n+k}^*$, then we get
    \begin{align}\label{eq.lotcterm1-2}
    \E(\tilde Z_n\tilde Z_{n+k}|Z_n) &= \E(\tilde Z_n (Z_n-\tilde Z_n)\E(\tilde Z_{n+k}^*|\tilde Z_n)|Z_n)\nonumber\\ &= \E\left(\tilde Z_n (Z_n-\tilde Z_n) p_{n+k} \frac{m_{n+k}}{m_n}\frac{1}{1-p_n}\Big|Z_n\right)\nonumber\\
    &= p_{n+k} \frac{m_{n+k}}{m_n}\frac{1}{1-p_n} \left(Z_n\E(\tilde Z_n|Z_n) - \var(\tilde Z_n|Z_n)-\E(\tilde Z_n|Z_n)^2\right)\nonumber\\
    &= p_{n+k} \frac{m_{n+k}}{m_n}\frac{1}{1-p_n} \left(Z_n^2p_n - Z_np_n(1-p_n) - Z_n^2p_n^2\right)\nonumber\\
    &= p_{n+k} \frac{m_{n+k}}{m_n}p_n Z_n.
    \end{align}
    Inserting \eqref{eq.lotcterm1-1} and \eqref{eq.lotcterm1-2} into \eqref{eq.lotcterm1}, we get that
    \begin{equation}\label{eq.lotcterm1-3}
    \cov(\tilde Z_n, \tilde Z_{n+k}|Z_n) = p_{n+k} \frac{m_{n+k}}{m_n}p_n Z_n - Z_np_np_{n+k}\frac{m_{n+k}}{m_n}Z_n = -p_{n+k}p_n\frac{m_{n+k}}{m_n} Z_n.
    \end{equation}
    Taking the expectation, we finally get that the first term on the right hand side of \eqref{eq.lotc} is given by
    \begin{equation}\label{eq.lotcterm1-4}
    \E\left(\cov(\tilde Z_n, \tilde Z_{n+k}|Z_n)\right) = \E\left(-p_{n+k}p_n\frac{m_{n+k}}{m_n} Z_n\right) = -p_{n+k}p_nm_{n+k}.
    \end{equation}
    Using \eqref{eq.lotcterm1-1}, we get that the second term on the right hand side of \eqref{eq.lotc} is given by 
    \begin{equation}\label{eq.lotcterm2}
    \cov\left(\E(\tilde Z_n|Z_n),\E(\tilde Z_{n+k}|Z_n)\right) = \cov\left(Z_np_n,p_{n+k}\frac{m_{n+k}}{m_n}Z_n\right) = p_np_{n+k}\frac{m_{n+k}}{m_n}s_n^2.
    \end{equation}
    Combining \eqref{eq.lotcterm1-4} and \eqref{eq.lotcterm2} and simplifying, we get \eqref{eq.tcovnk}.
    Finally, \eqref{eq.tmeanall} and \eqref{eq.tvarall} follow by the usual formulas for means and variances of sums, followed by simplifying the expressions.
\end{proof}

\subsection{Probability generating functions and moments on tree height}

Denote the probability generating function (pgf) corresponding to the offspring distribution $f_n$ by $g_n(s) = \sum_{i=0}^\infty s^i f_n(i)$ for $s\in[0,1]$. Then Proposition~1 in \cite{fearn72} shows that we can obtain the pgf for the number of vertices of generation $n=1,2,\ldots$ by  
    \[
	G_n(s) = g_0 \circ g_1 \circ \dots \circ g_{n-1}(s),
	\]
where $\circ$ denotes composition. From this we can obtain the cumulative distribution function for the highest generation number of a vertex in an IGWT (i.e.\ the tree height not counting generation 0) as
	\[
		P(N_{\max}\leq n) = P(Z_{n+1}=0) = G_{n+1}(0),
	\]
and from this we also have the mean and variance given by
\begin{align*}
\E(N_{\max}) &= \sum_{n=0}^\infty \left(1-G_{n+1}(0)\right), \\
\var(N_{\max}) &= \left(\sum_{n=0}^\infty (2n+1)\left(1-G_{n+1}(0)\right)\right) - \left(\sum_{n=0}^\infty \left(1-G_{n+1}(0)\right)\right)^2.
\end{align*}
Note that $G_n$ can only be derived analytically in very simple cases, but it can easily be obtained numerically.


\subsection{Sufficient condition for finiteness}\label{sec.finite}

In order to create a realistic model for distribution networks, it is natural to require that the model generates networks with a finite number of vertices, something that is not ensured by Definition~\ref{def.igwt}. There are various papers (e.g.\ \cite{agresti75} and \cite{fearn81}) considering different criteria for finiteness of inhomogeneous Galton-Watson processes, but here we focus on a simple and sufficient (but not necessary) condition for finiteness, which is adequate to cover most relevant cases.

\begin{proposition}\label{prop.finite}
    If there exist some $\tilde n$ and $c$ such that $\mu_n\leq c<1$ for all $n>\tilde n$, then $m<\infty$, and hence the number of vertices are finite with probability one.
\end{proposition}
\begin{proof}
    Assuming that $\tilde n$ and $c$ exist such that $\mu_n\leq c<1$ for all $n>\tilde n$, we get from \eqref{eq.meanall} that
    \begin{align*}
    m &= \sum_{n=0}^\infty \prod_{i=0}^{n-1} \mu_i\\ &\leq \sum_{n=0}^{\tilde n} \prod_{i=0}^{n-1} \mu_i + \left(\prod_{i=0}^{\tilde{n}-1}\mu_i\right) \sum_{n=\tilde{n}+1}^\infty  c^{n-(\tilde n +1)}\\
    &= \sum_{n=0}^{\tilde n} \prod_{i=0}^{n-1} \mu_i + \left(\prod_{i=0}^{\tilde{n}-1}\mu_i\right) \frac{1}{1-c},
    \end{align*}
    which is finite.
\end{proof}

Note that changing the inequalities in Proposition~\ref{prop.finite} to get $\mu_n\geq c\geq1$ for all $n>\tilde n$, we can similarly prove that $m=\infty$ (provided no $\mu_i=0$). While this does not necessarily mean that there is a positive probability for obtaining an infinite network, it still means that we will occasionally get unrealistically massive networks, suggesting that a model fulfilling this will not be realistic.

\section{Statistical models based on IGWTs}\label{sec.models}

\subsection{Modeling framework}\label{sec.modframe}

In order to formulate a specific IGWT, we only need to specify the offspring distributions $f_n$. Since so many results only rely on means and variances, we will restrict ourselves to considering $f_n$ that can be parametrised uniquely by its mean and variance, and express the dependence of $f_n$ on $n$ through two functions expressing the mean and variance structures. With a slight abuse of notation, we write
\begin{equation}\label{eq.meanvarstruc}
\mu_n = \mu_n(\phi), \qquad \sigma_n^2 = \sigma_n^2(\psi),
\end{equation}
where $\phi$ and $\psi$ are parameter vectors.  Although \eqref{eq.meanvarstruc} allows for the mean and variance structures to be set independently by selecting non-negative functions, not all settings provide valid mean and variance combinations; we will return to this issue in Section~\ref{sec.offspring}.

\begin{example}
    We give a few examples of possible mean and variance structures that we will use in the analysis in Section~\ref{sec.analysis}; we write them as mean structures, but obviously they can also be used as variance structures:
    \begin{align} 
    \mu_n(\phi) &= \phi_1 \phi_2^n, \qquad \phi = (\phi_1,\phi_2) \in [0,\infty)^2,\label{eq.exp}\\
    \mu_n(\phi) &= \begin{cases} \phi_1, &\quad n=0 \\ \phi_2 \phi_3^n, &\quad n>0\end{cases},
    \qquad \phi = (\phi_1,\phi_2,\phi_3) \in [0,\infty)^3,\label{eq.aexp}\\
    \mu_n(\phi) &= \begin{cases} \phi_1, &\quad n=0 \\ \phi_2 \phi_3^n+\phi_4, &\quad n>0\end{cases},
    \qquad \phi = (\phi_1,\phi_2,\phi_3,\phi_4) \in [0,\infty)^4.\label{eq.aexppc}
    \end{align}
    Note that by Proposition~\ref{prop.finite} in order to ensure that the resulting network is finite with probability one, the parameters should be restricted. For example, $\phi_2<1$ for \eqref{eq.exp}, $\phi_3<1$ for \eqref{eq.aexp}, and $\phi_3<0$ and $\phi_4<1$ for \eqref{eq.aexppc} imply finiteness.
\end{example}

\subsection{Offspring distributions}\label{sec.offspring}

In addition to specifying the mean and variance structure, we also need to specify the offspring distributions for each generation. We assume that we fix the class of distribution for each generation (either with the same class or different classes of distributions), but let its parameters depend on the generation. However, to ease the notation in this section, we suppress the dependence of all parameters on the generation, i.e.\ we do not include the subscript $n$ on the parameters.

Any distribution on the non-negative integers may be used as offspring distribution. However, the modeling framework in Section~\ref{sec.modframe} requires the distribution to be parametrized by its mean and variance. This can be a one-parameter distribution, such as a Poisson or geometric distribution, or something even simpler, such as a distribution on $\{0,2\}$. In this case we should only specify one of the mean or the variance structures. Or it could be a two-parameter distribution, such as a binomial or negative binomial distribution.

However, for the data in Section~\ref{sec.analysis}, we need the distribution to reflect the fact that no vertices in the data have exactly one offspring. One simple option for making a distribution on $\mathbb N_0\backslash\{1\}$ is to make a mixture between 0 and a discrete distribution on $\mathbb{N}_0$ translated by 2. With a suitable choice of a standard one-parameter discrete distribution, we get a two-parameter mixture distribution, which can be parametrized by its mean and variance. Two examples follow:

\begin{example}[Poisson-zero mixture model]\label{ex.pois}
    One simple option is to use the Poisson distribution translated by 2 in the mixture model. We define this by the probability function 
    \begin{equation*}
        f(x; p,\lambda) = \begin{cases}
            p, & x=0\\
            0, & x=1\\
            (1-p)\frac{e^\lambda \lambda^{x-2}}{(x-2)!}, & x\geq 2
        \end{cases} 
    \end{equation*}
    with parameters $p\in[0,1]$ and $\lambda>0$. By the usual formulas for obtaining mean and variance for mixture distributions, the mean and variance are given by
    \begin{equation}\label{eq.poismom}
    \mu = (1-p)(\lambda+2),\qquad 
    \sigma^2 = (1-p) \left(\lambda+p(\lambda+2)^2\right).
    \end{equation}
    Isolating $p$ and $\lambda$ from these through straightforward, but tedious, calculations gives
    \begin{equation}\label{eq.poistrans}
    \lambda = \frac{1}{2} \left(\frac{\sigma^2}{\mu}+\mu-1 + \sqrt{\left(1-\frac{\sigma^2}{\mu}-\mu\right)^2+8}\right)-2, \qquad p = 1-\frac{\mu}{\lambda+2}.
    \end{equation}
    The pgf for the offspring distribution is given by
    \[
    g(s) = p + (1-p)s^2 e^{\lambda(s-1)}.
    \]
\end{example}

\begin{example}[Geometric-zero mixture model]\label{ex.geom}
    Similarly to Example~\ref{ex.pois} we can make a mixture of 0 and a geometric distribution translated by 2. This is defined by
    \begin{equation*}
        f(x; p,q) = \begin{cases}
            p, & x=0\\
            0, & x=1\\
            (1-p)(1-q)^{x-2}q, & x\geq 2
        \end{cases} 
    \end{equation*}
    with parameters $p,q\in[0,1]$. The mean and variance is given by
    \[
    \mu = (1-p)\left(\frac{1-q}{q} + 2\right),\qquad \sigma^2 = (1-p)\frac{1-q}{q^2} + \frac{1}{1-p}\mu^2,
    \]
    and the reverse relations are given by
    \begin{equation}\label{eq.geomtrans}
    p = 1 - \frac{1}{4}\left(\sigma^2+3\mu+\mu^2 - \sqrt{(\sigma^2+3\mu+\mu^2)^2-16\mu^2}\right), \qquad q = \left(\frac{\mu}{1-p}-1\right)^{-1}           
    \end{equation}
    The pgf is given by
    \[
    g(s) = p + \frac{(1-p)s^2 q}{1-(1-q)s}.
    \]
\end{example}

For both distributions we can uniquely transform between parametrising them by the parameters used in the definition and by the mean and variance. However, not all combination of $(\mu,\sigma^2)\in\mathbb{R}_+^2$ are possible; essentially there is a minimum variance that we can get, which depends on the mean. This is true for any discrete distribution, where we can only have a variance of zero for integer mean values, but for Examples~\ref{ex.pois} and \ref{ex.geom} there are further restrictions on the minimum variance.

For the Poisson-zero mixture model, the minimum possible variance for $\mu\in[0,2]$ is obtained for $\lambda=0$, and from \eqref{eq.poismom} and \eqref{eq.poistrans} it follows that $\mu=2(1-p)$ and $\sigma^2 = (1-p)4p$, and thus $\sigma^2 = 2\mu (1-\mu/2)$. For $\mu>2$, the minimal variance is obtained for $p=0$, yielding $\sigma^2=\mu-2$. Similarly lower bounds on the variance can be obtained for the geometric-zero mixture model. The possible mean-variance combinations for both models are shown in Figure~\ref{fig.meanvarcomb} together with the means and variances estimated directly from each generation combined for all networks in the data. Generation 14 is outside the possible regions for both models, but only one network out of the 100 actually has this many generations, so this is not practically relevant. Generation 0 on the other hand is very relevant, and this is only possible for the Poisson-zero mixture model. We return to this issue in Section~\ref{sec.modspec}. 

\begin{figure}
    \centering
    \includegraphics[width=0.7\linewidth]{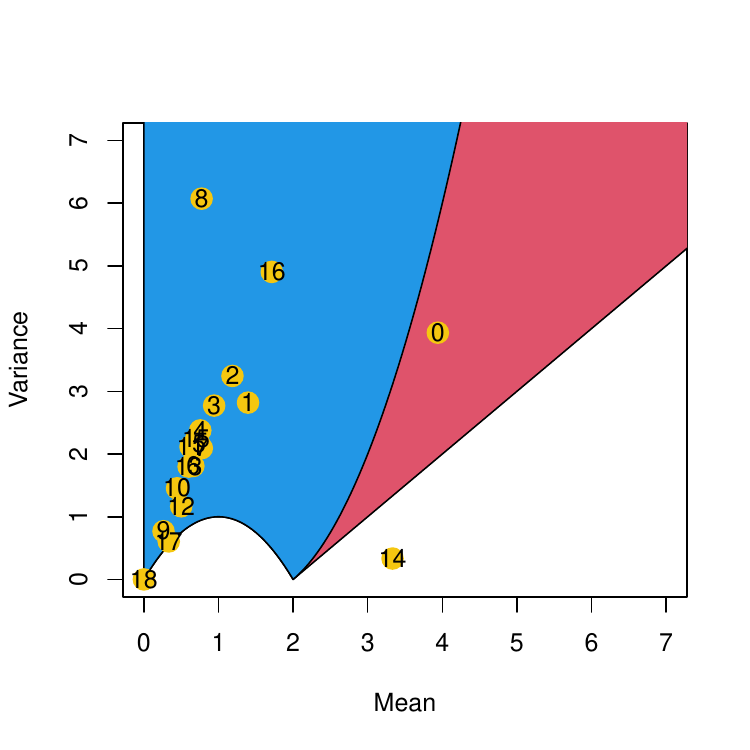}
    \caption{The possible mean-variance combinations for the Poisson-zero mixture model (red and blue) and the geometric-zero mixture model (blue). The yellow dots show the mean-variance combinations for the offspring distributions estimated for each generation in all networks in the data combined, where the number indicates the generation.
    }
    \label{fig.meanvarcomb}
\end{figure}

It should be noted that in order to keep the model simple, it is natural to keep the type of offspring distribution fixed and only let the parameters depend on the generation. However, the data may show a better fit if we use different offspring distributions for different generations, and given that our modeling framework parametrizes everything uses the mean and variance, there is no added difficulty in using varying offspring distributions.





\section{Parameter Estimation}\label{sec.mle}

For parameter estimation we will use maximum likelihood estimation. 

Assume that we have a single tree-shaped network available, say $T$, where we know all vertices and directed edges. Denote the highest generation number in $T$ by $n_{\max}$ (i.e.\ the height of the tree not counting the root), and the number of vertices in generation $n$ by $z_n$. Furthermore, denote the number of offspring of vertex $i$ in generation $n$ by $x_{n,i}$. Note that we will assume that the $x_{n,i}$ offspring of vertex $i$ are numbered $1,\ldots,x_{n,i}$, i.e.\ $T$ is an ordered tree.

Now assume that we have chosen a specific class of offspring distributions $f_n$ specified by chosen mean and variance structures $\mu_n(\phi)$ and $\sigma_n^2(\psi)$, where we concatenate the two parameters vectors into one vector $\theta = (\phi,\psi)$. Then the likelihood function can be written explicitly as
\begin{equation}\label{eq.lik}
L(\theta;T) = \prod_{n=0}^{n_{\max}}\prod_{i=1}^{z_n} f_n(x_{n,i};\theta).
\end{equation}
The explicit form permits analytical maximization only in very simple cases; more complex cases may be solved numerically, for example using Newton-Raphson.


If we have multiple networks available and wish to fit a single model based on all the networks, then, assuming that the networks are independent, we simply need to multiply multiple likelihood functions given by \eqref{eq.lik}, again resulting in a likelihood function of the same form. So in practice the only difference in estimation based on multiple networks is a larger dataset.

If we consider one of the mixture models from Example~\ref{ex.pois} or Example~\ref{ex.geom}, we then get the log likelihood
\begin{align}
\log L(\theta;T) &= \sum_{n,i: x_{n,i}=0} \log(f_n(0)) + \sum_{n,i: x_{n,i}\geq 2} \log(f_n(x_{n,i})) \nonumber\\
&= \sum_n N_{n,0} \log(p_n) + \sum_n \sum_{k\geq2} N_{n,k} \log(f_n(k)),\label{eq.lik2}
\end{align}
where $N_{n,k}$ denotes the number of generation $n$ vertices with exactly $k$ offspring and $p_n$ denotes the probability that a generation $n$ vertex has 0 offspring. It is worth noticing that while numerical maximization may require many evaluations of the log likelihood function, the primary computation time comes from counting all $N_{n,k}$; this only has to be done once, making numerical maximization very fast if implemented properly.

Inserting $f_n$ from the Poisson-zero mixture model in Example~\ref{ex.pois} into \eqref{eq.lik2} and omitting terms that do not depend on the parameters, we get 
\[
\log L(\theta;T) = \sum_n N_{n,0} \left(\log(p_n) + A_n(\log(1-p_n)-\lambda_n) + B_n\log(\lambda_n)\right)
\]
where $\lambda_n$ denotes the $\lambda$-parameter for generation $n$ and
\[
A_n = \sum_{k\geq2} N_{n,k}, \qquad B_n = \sum_{k\geq2} N_{n,k}(k-2).
\]
Similarly we get for the geometric-zero mixture model that
\[
\log L(\theta;T) = \sum_n N_{n,0} \left(\log(p_n) + A_n(\log(1-p_n)+\log(q_n)) + B_n\log(1-q_n)\right),
\]
where $q_n$ is the $q$-parameter for generation $n$. Here we should remind the reader that explicit dependence of $\log L$ on $\theta$ is given by relating $p_n$ and $\lambda_n$ or $q_n$ to $\mu_n$ and $\sigma_n^2$ using \eqref{eq.poistrans} or \eqref{eq.geomtrans}, and then relating these to $\theta=(\phi,\psi)$ using the chosen mean and variance structure \eqref{eq.meanvarstruc}.

\section{Statistical analysis of distribution networks}\label{sec.analysis}

\subsection{Data from Danish low voltage electricity distribution grids}

Next we turn to creating a IGWT model using the framework from Section~\ref{sec.models} to model a distribution network. As data we have 100 Danish networks available, three of which are shown in Figure~\ref{fig.3grids}. The full dataset can be found in \cite{dataset}.
The data contains the tree topology all the networks, and it has been cleaned in the following way:
\begin{itemize}
    \item All networks had a single edge going from the substation to the next vertex. The substation has been removed, so the root now represents the first connection to the substation. This for example means that networks can consist of just one node, which can then be interpreted as a single consumer, such as a large factory, as the only vertex connected to the substation (there are seven such networks in the dataset). 
    \item All intermediate vertices having only one outgoing edge (i.e. a junction box with just one outgoing connection) have been removed, and its parent and offspring have been joined instead. This has been done to focus on shape of the network. A consequence of this is that no vertex has exactly one offspring, hence motivating the use of the offspring distributions given in Examples~\ref{ex.pois} and \ref{ex.geom}.
\end{itemize}
Clearly, the particular data cleaning procedure affects the data. Applying a different cleaning procedure may necessitate revision of the particular model specified in Section~\ref{sec.modspec}.


\subsection{Model specification}\label{sec.modspec}

\begin{figure}
    \centering
    \includegraphics[width=0.9\linewidth]{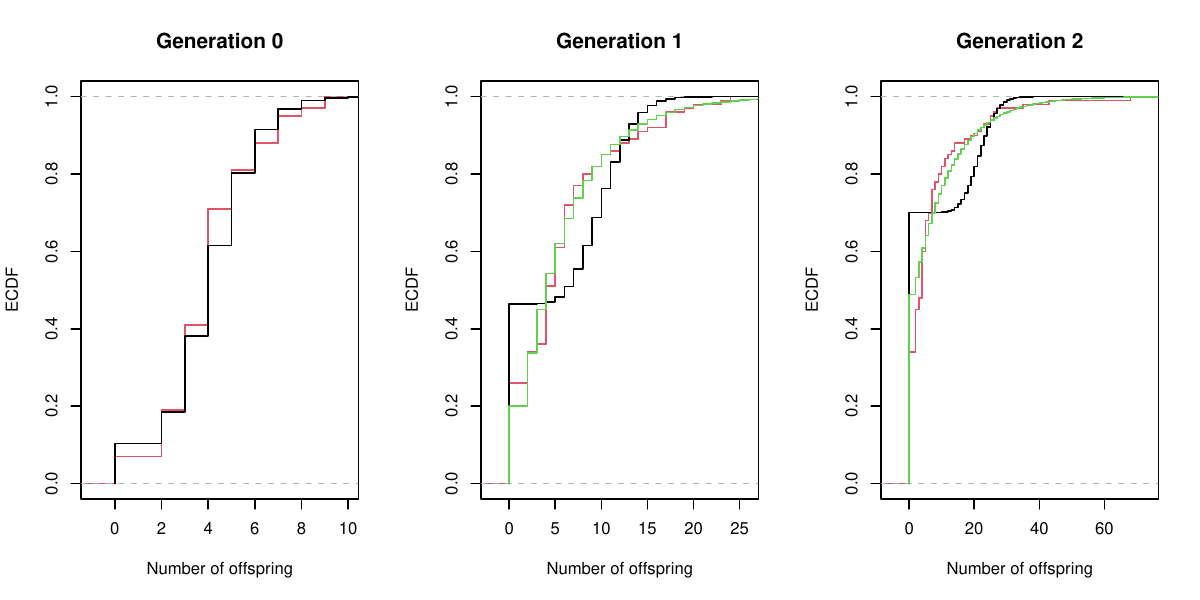}
    \caption{Empirical cumulative distribution functions for generations 0, 1, and 2 (red), along with Poisson-zero mixture distribution (black) and geometric-zero mixture distribution (green), both fitted to data by estimated mean and variance. The green curve is missing for generation 0, since no geometric-zero mixture distribution exists for the estimated mean and variance.}
    \label{fig.offdist}
\end{figure}

Firstly, we select an offspring distribution. Figure~\ref{fig.offdist} reports the empirical cumulative distribution function for the offspring distributions of generation 0, 1, and 2 in the data. Also the Poisson-zero and geometric-zero mixture distributions with parameters estimated by matching mean and variance to the data are shown in order to visualize which distribution potentially fits the best. Generation 0 is missing the curve for the geometric-zero mixture distribution, since there exists no such distribution with the combined mean and variance estimated from the data (see Section~\ref{sec.offspring} for details on this). This suggests that we use the Poisson-zero mixture distribution for generation 0. On the other hand, the geometric-zero mixture distribution seems to fit the data much better than the Poisson-zero mixture distribution for generations $n>0$, so we use the geometric-zero mixture distribution for those generations (plots for generations $n>2$ have been omitted from Figure~\ref{fig.offdist}, but show similar tendencies as $n=1$ and 2).

Secondly, we have to decide on mean and variance structures. Figure~\ref{fig.meanvaroff} shows the estimated means and variances of the offspring distributions for each generation in the data. Since only one network in the data goes beyond generation 12, we can more or less ignore the highest generations. Apart from generation 0, the means visually resembles an exponential function. 
However, by trial and error the fit is better if we add a constant, thus motivating the mean structure given by
\begin{equation*}
\mu_n(\phi) = \begin{cases} \phi_1, &\quad n=0 \\ \phi_2 \phi_3^n+\phi_4, &\quad n>0\end{cases},\qquad \phi = (\phi_1,\phi_2,\phi_3,\phi_4) \in [0,\infty)^4.
\end{equation*}
The right panel of Figure~\ref{fig.meanvaroff} resembles an exponential decrease, so we let the variance structure be given by
\[
\sigma_n^2(\psi) = \psi_1 \psi_2^n, \qquad \psi = (\psi_1,\psi_2) \in [0,\infty)^2.
\]
We have considered a number of other mean and variance structures, but the above gave the best fit by visual inspection and comparison of various summaries (such as those used for model checking in Section~\ref{sec.modcheck}). A more careful analysis, comparing various models for example using Akaike's information criterion, or other model selection criteria, could also be considered for obtaining a well-fitting model, but we will not pursue that direction in the current paper.

\subsection{Estimation and interpretation of parameters}

We estimate the parameters of the model by maximum likelihood estimation based on all 100 networks from the data.
Maximizing the likelihood function, we get the estimates
\[
\hat\phi = (3.94, 1.16, 0.654, 0.613), \qquad
\hat\psi = (3.35, 0.958).
\]
The estimated mean and variance structures are shown as black curves in Figure~\ref{fig.meanvaroff}. To ease comparison to the data, 95\% confidence intervals have also been added (the reason for the seemingly missing confidence interval in both plots at the last generation is that there are only four vertices of this generation in the data, and they happen to have the same number of offspring). The mean structure seems to fit very well with the data (disregarding the few outliers at high generations, which anyway represents data coming from just a single grid). The fit of the variance structure seems decent, although the figure reveals some discrepancies, where the fitted curve is outside the confidence intervals.
\begin{figure}
    \centering
    \includegraphics[width=0.9\linewidth]{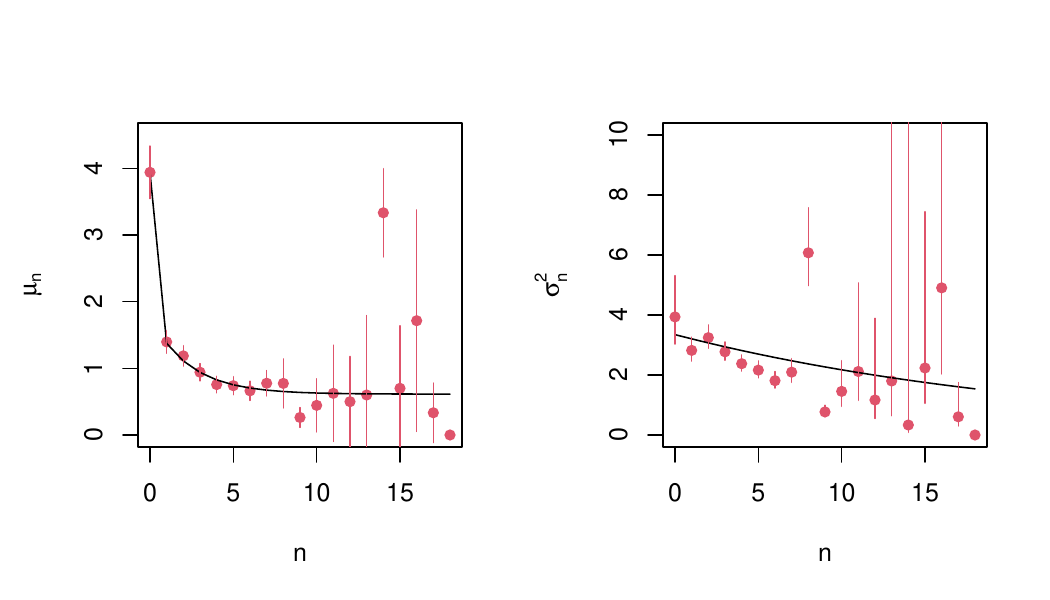}
    \caption{Estimated means and variances of the number of offspring for each generation obtained from the data (red dots), and mean and variance structures from the estimated model (black curves). The vertical red lines are the 95\% confidence intervals estimated from the data.}
    \label{fig.meanvaroff}
\end{figure}

Since the model is fairly complex, in the sense that it has six parameters, it makes sense to spend a moment interpreting these parameters.
\begin{itemize}
    \item $\phi_1$: This parameter together with $\psi_1$ give the mean and variance of the Poisson-zero mixture offspring distribution for generation 0. Using \eqref{eq.poistrans} the estimates can be translated to estimates of the probability of no offspring $\hat p_0 = 0.0751$ and the parameter used in the Poisson distribution $\hat\lambda_0 = 2.257$. Note that $\hat p_0$ is the probability of a network only containing one node while $\hat\lambda_0+2 = 4.257$ is the mean number of offspring given that there are offspring.
    \item $\phi_2$ and $\phi_3$: The parameter $\phi_3$ controls the exponential increase or decrease (ignoring the additive term $\phi_4$) of the number of offspring for each generation, while the product $\phi_2\phi_3$ gives the mean number of offspring in generation 1. Since the estimate $\hat\phi_3<1$, the number of offspring from each parent decreases with each generation. Note that in general high values of $\phi_3$ and low values of $\phi_2$ will result in high, but thin, trees, and the opposite will results in low, but wide, trees. Also note that if either $\phi_2=0$ or $\phi_3=1$, the mean structure becomes constant for $n>0$, but the estimates indicate that this is far from the case (a rigorous test based on the asymptotic distribution of the estimates could be developed, but this is beyond the scope of the paper). 
    \item $\phi_4$: This is an additive constant for the mean of all generations $n>0$. Using that the estimate $\hat\phi_4<1$ together with $\hat\phi_3<1$, we get by Proposition~\ref{prop.finite} that the estimated model produces networks that are finite with probability one.
    \item $\psi_1$ and $\psi_2$: The parameter $\psi_1$ gives the variance of the number of offspring from generation 0, while the parameter $\psi_2$ gives the rate of the exponential increase or decrease of the variance for each generation. Since the estimate $\hat\psi_2$ is below, but very close to, one, the estimated model has a rather constant variance with a slow decrease. Since $\hat\psi_2$ is so close to one, it is tempting to ask the question of whether we could simply have used a constant variance structure without changing the results significantly. However, we did try the model with a constant variance, and it turns out the results are rather different. 
\end{itemize}

In addition to the parameter estimates, we can also interpret the fact that the best fit was obtained by using a Poisson-zero mixture for generation 0 and a geometric-zero mixture model for the other generations. Essentially the Poisson distribution has a lighter tail than the geometric distribution, and also has the possibility of a mode placed away from the lowest value, so this choice produces a more stable number of generation 1 vertices. The geometric-zero mixture offspring distribution used for the other generations means that there are more leaves, but also more vertices with many offspring, than a corresponding Poisson distribution would imply. In other words, the higher generations have a more varying structure.

\subsection{Model checking}\label{sec.modcheck}

The fitted model can be used to simulate network realizations as those shown in Figure~\ref{fig.3gridssim}. A first visual comparison to the data in  Figure~\ref{fig.3grids} reveals no obvious discrepancies.
\begin{figure}
    \centering
    \includegraphics[width=0.9\linewidth]{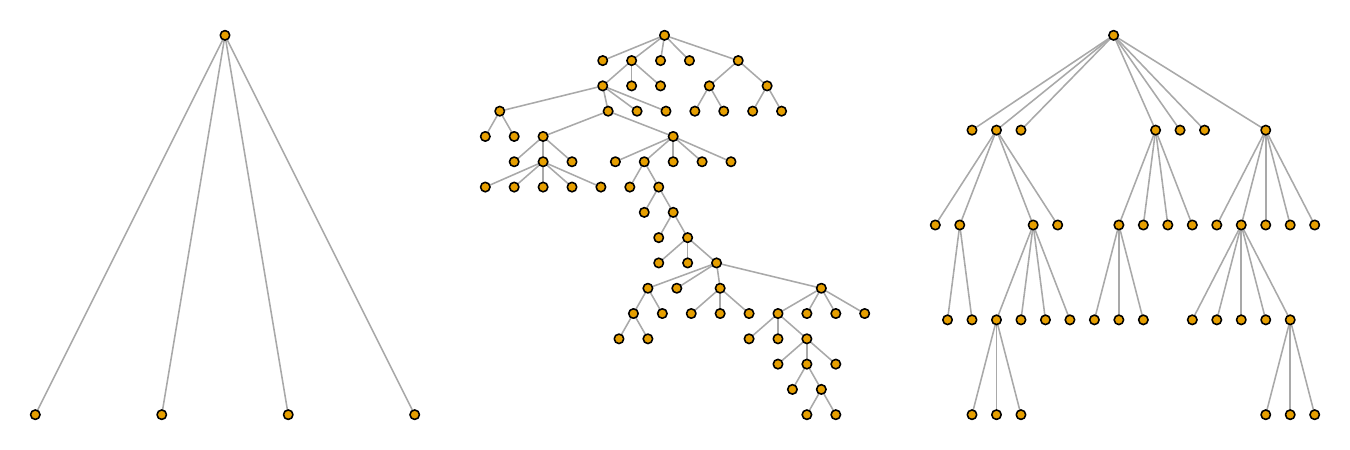}
    \caption{Three realizations from the fitted model.}
    \label{fig.3gridssim}
\end{figure}

The fit of the model can be studied  more  thoroughly by means of summary statistics. Specifically, we compare total number of vertices $Z$, height of tree $N_{\max}$ and number of leaves $\tilde Z$ calculated for the data to those obtained for the model by   simulations or by analytical calculations. For the simulation based approximations, we use 10000 runs. Means and variances for the summaries are reported in  Table~\ref{tab.summaries}.
  
In Figure~\ref{fig.totverecdf} we compare empirical cumulative distribution functions for the total number of vertices $Z$ of the model simulations and data. In the figure we have added approximate (i.e.\ based on asymptotics) point-wise 95\%-confidence intervals for the data. The functions are fairly similar, although the data has more mass at low values and a heavier tail than the simulations. This is reflected in the means and variances reported in Table~\ref{tab.summaries}. As to be expected, the analytical results agrees closely with the simulation. We observe that the model fits the mean of the data  excellently, but the variance of the model is about half of that of the data.

\begin{figure}
    \centering
    \includegraphics[width=0.4\linewidth]{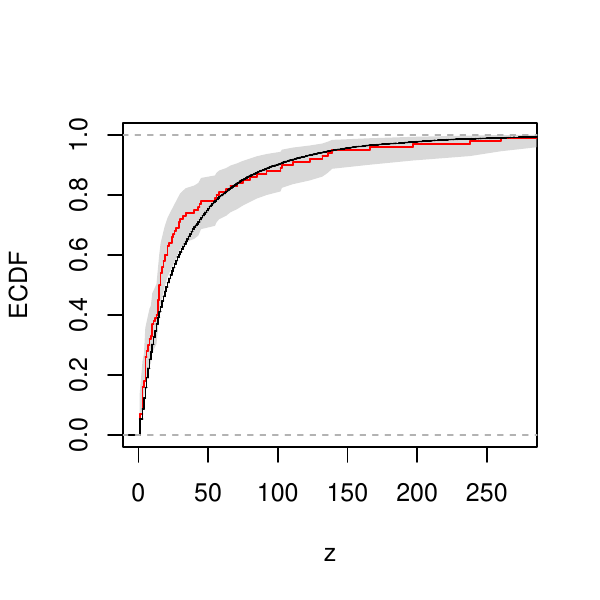}
    \caption{Empirical cumulative distribution functions of the total number of vertices in simulated networks using the fitted model (black) and networks in the data (red). The grey region shows approximate 95\% point-wise confidence intervals for the data.}
    \label{fig.totverecdf}
\end{figure}
\begin{table}[]
    \centering
    \begin{tabular}{lccc}
    \toprule
 Summary statistic   & \multicolumn{3}{c}{(Mean, Variance)} \\
 \cmidrule(lr){2-4}
     &Data  &   Simulation & Analytical  \\
    \midrule
    Total number of vertices, $Z$ &(37.7, 3171)& (37.5, 1743)& (37.2, 1783)\\[0.3ex]
    Height of tree, $N_{\max}$ & (4.50, 12.2) & (4.88, 12.0) & (4.85, 11.9)\\[0.3ex]
    Number of leaves, $\tilde Z$ & (26.5, 1756) & (26.5, 917) & (26.4, 913)\\
    \bottomrule
    \end{tabular}
    \caption{Summaries for the data and fitted model. For the model, both simulated and analytically obtained values are reported.}
    \label{tab.summaries}
\end{table}
To understand the discrepancy in more detail, we plot mean and variance of the number of nodes per generation  in   Figure~\ref{fig.genscat}.
Again, the mean fits quite nicely, while the model underestimates the variance primarily in generations 3 to 5. Given that the mean and variance of the offspring distributions fit well between the model and the data (see Figure~\ref{fig.meanvaroff}) and Proposition~\ref{prop.momn} shows that the mean and variance only depends on the mean and variance of the offspring distribution under the assumptions of the IGWT, this reveals that these assumptions do not quite fit with the data. We will analyse this discrepancy more thoroughly in Section~\ref{sec.lackoffit}.
\begin{figure}
    \centering
    \includegraphics[width=0.9\linewidth]{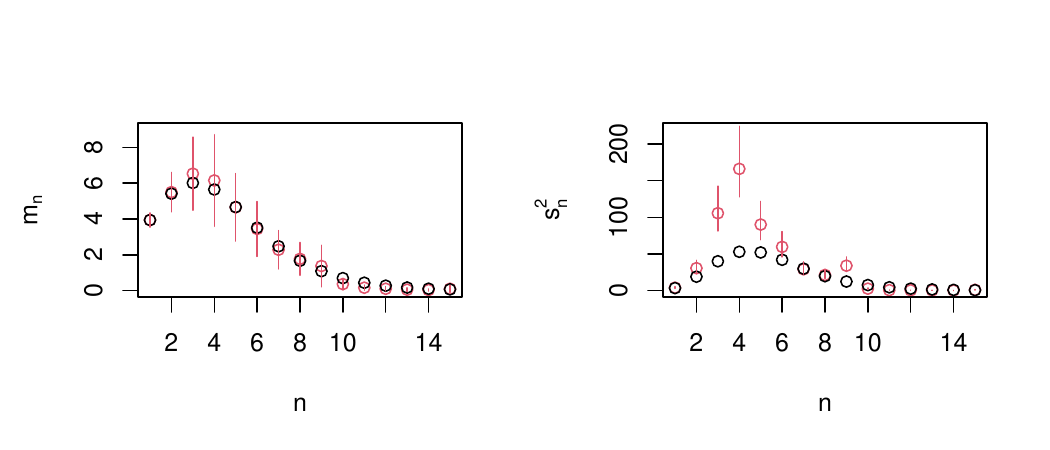}
    \caption{The mean and variance for the number of vertices in each generation 1 to 15 calculated analytically (black) and approximated from the data (red). The vertical red bars show the 95\% confidence intervals estimated from the data.}
    \label{fig.genscat}
\end{figure}

The mean and variance of tree height $N_{\max}$ are shown in   Table~\ref{tab.summaries}.
The model and the data are very similar. This is also reflected in the empirical cumulative distribution function reported in Figure~\ref{fig.totgenleavesecdf}.

\begin{figure}
    \centering
    \includegraphics[width=0.4\linewidth]{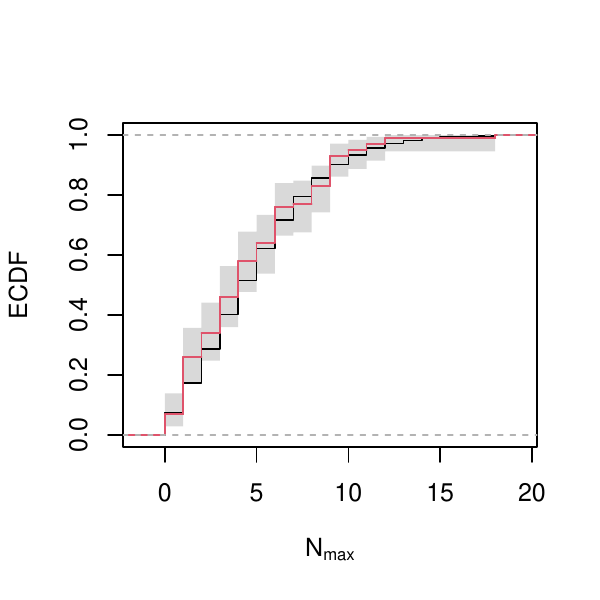}
    \includegraphics[width=0.4\linewidth]{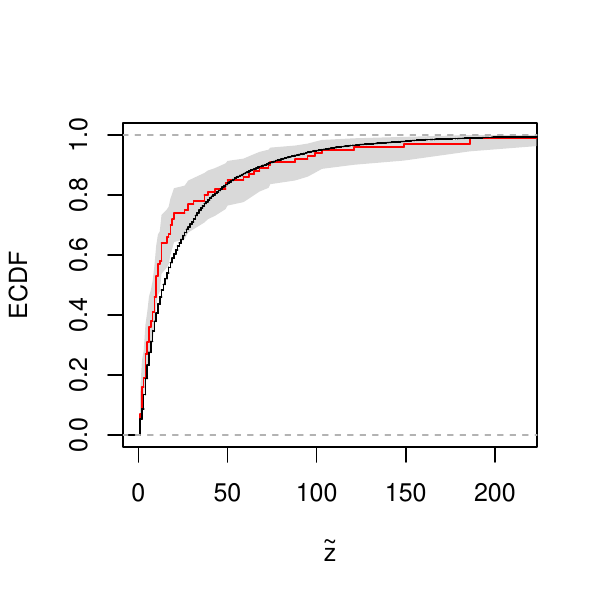}
    \caption{Left panel: The cumulative distribution function of the tree height for the estimated model calculated analytically (black), and the empirical cumulative distribution function estimated from the data (red). The grey area shows point-wise 95\%-confidence intervals for the data. Right panel: Empirical cumulative distribution functions showing the number of leaves for each generation of the simulations (black) and the data (red). The grey area shows point-wise 95\%-confidence intervals for the data.}
    \label{fig.totgenleavesecdf}
\end{figure}

Finally, we consider the number of leaves $\tilde{Z}$. Table~\ref{tab.summaries}  reveals the same problem as for the total number of vertices, i.e.\ the mean fits nicely between the simulations and the data, but the variance of the model is about half  of that of the data. This observation is in accordance with the empirical distribution functions in 
Figure~\ref{fig.totgenleavesecdf}. 




\subsection{Lack of fit of variances}\label{sec.lackoffit}

The estimated model seems to fit the data quite well, except that the variance on the number of vertices in each generation is underestimated (the other underestimated variances is a consequence of this). Since the mean and variance of the offspring distributions both fit well, the variance can only occur because the assumptions in Definition~\ref{def.igwt} are not fulfilled. The assumed independence is the likely problem, since a positive correlation between some or all of the offspring distributions could result in higher variances on the number of offspring in each generation. Real networks would tend to have many or few vertices depending on where it is located: A network located in a residental area with many houses would often have many consumers, while a network in an industrial area there might just be a few factories as consumers. This would create positive correlation between the number of vertices. 

To make a quick investigation of the potential consequences of correlation, we will consider an example where we focus on generations 1 and 2. Assume that all the offspring distributions of generation 1 have a positive correlation, say $\rho$, if they belong to the same network (but are still independent for different networks). If offspring distributions of generation 1 are thus correlated, the variance of the number of vertices in generation 2 becomes
\begin{align*}
\var\left(\sum_{i=1}^{X_0}X_{1,i}\right) 
&= \E\left(\var\left(\left.\sum_{i=1}^{X_0}X_{1,i}\right|X_0\right)\right)
+ \var\left(\E\left(\left.\sum_{i=1}^{X_0}X_{1,i}\right|X_0\right)\right)\\
&= \E\left(\sum_i\var\left(X_{1,i}|X_0\right)
+\sum_{i\not=j}\cov\left(X_{1,i},X_{1,j}|X_0\right)\right)
+ \var\left(\sum_i\E(X_{1,i}|X_0)\right)\\
&=\E(X_0\sigma_1^2 + (X_0^2-X_0)\sigma_1^2\rho)
+ \var(X_0\mu_1)\\
&=\sigma_1^2\mu_0 + (\sigma_0^2+\mu_0^2-\mu_0)\sigma_1^2\rho + \sigma_0^2\mu_1^2
\end{align*}
Comparing with $s_2^2=\sigma_0^2\mu_1^2+\sigma_1^2\mu_0$ given by Proposition~\ref{prop.momn}, we can see that the variance grows by $(\sigma_0^2+\mu_0^2-\mu_0)\sigma_1^2\rho$. To see the possible consequences of this difference on the model, we take values for $\sigma_0^2$, $\sigma_1^2$, $\mu_0$, and $\mu_1$ corresponding to the estimates from the data, and insert the extreme cases $\rho=0$ or $\rho=1$ to obtain $s_2^2=18.8$ and $s_2^2=62.5$, easily covering the variance estimate from the data given by $30.6$. We stop here and merely conclude that it is possible that the introduction of such a correlation may provide a model with a better fit of the variance, leaving a full analysis to future work.

The are of course many possibilities for introducing correlation into an IGWT. The following list gives a few examples (where the first one corresponds to the above):
\begin{enumerate}
    \item correlation within generation,
    \item correlation between offspring of common parent,
    \item dependence on parent,
    \item dependence on latent variable for the entire network.
\end{enumerate}
We shall leave it to future work to see, if any of these can provide a model that fits the data better, and maybe is able to fully capture the variation in the number of vertices. 

\section{Concluding remarks}\label{sec.con}

The IGWT with the modeling framework in Section~\ref{sec.modframe} using mean and variance structures gives a flexible model class for modeling distribution networks without loops. Obviously, such models may also be employed for other kinds of networks than electrical distribution networks, provided there are no loops. The maximum likelihood estimation method from Section~\ref{sec.mle} crucially relies on the network being completely known. If some connections are unknown, or even worse, we only know the total number of vertices in each generation, the estimation can be considered as a missing data problem, for example using the expectation-maximization algorithm, but at considerable computational increase for the estimation. On the other hand, the computational burden of the simulation algorithm is not influenced significantly by this. 

Furthermore, Section~\ref{sec.igwt} provides various theoretical results on moments for number of vertices and number of leaves, both for individual generations and for the entire tree. The results for the individual generations are typically of a form that can be calculated explicitly, while the results for the entire tree requires approximations of the infinite sums, except in very simple cases. Also results are given on the moments of the height of the tree are given in a form that allows easy numerical approximations.

Turning our attention to the practical application, the paper provides a estimated model in Section~\ref{sec.analysis} that fits most aspects of the data, but with room for improvement. In particular, one of the four suggestions given in Section~\ref{sec.lackoffit} could improve the fit by extending the IGWT.

Another possibility for extending the models for this particular data could be to exchange the use of the Poisson and geometric distributions for a negative binomial distribution (which includes the other two as special or limiting cases). This would allow a flexible way of dealing with the difference in offspring distributions in generation 0 and the other generations. The disadvantage of this approach is that using the negative binomial distribution in the same type of mixture distribution as in Examples~\ref{ex.pois} and \ref{ex.geom} results in a three parameter distribution, so it cannot be parametrized by the mean and variance structure alone. This could be handled by extending the framework with a parametric formulation for another moment, such as skewness, or the probability of 0 offspring, which of course comes at the expense of increasing the number of parameters.

Finally, the following observation in Figure~\ref{fig.meanvarcomb} motivates a modification of the modeling framework: Most of the points (except generations 0, 8, 14) are located close to a straight line, so it could be possible to model the variance structure as a linear function of the mean structure. Generation 0 would then need to have a separate parameter for its variance, while it probably is possible to ignore the lack of fit for the other two generations, since generation 8 is not far off the line, and generation 14 only occurs in one network. An explicit use of the mean-variance relation could also motivate the use of generalized linear models, since a central component of specifying such a model is exactly this relation, which is known as the variance function in this context.


\acks 
We would like to thank the anonymous Danish DSO, who provided us with data for the analysis.

\fund 
The authors gratefully acknowledge funding from the Eurostars-EUREKA Programme under the ReDistXAI project. ReDistXAI is co-funded by Innovation Fund Denmark and the European Union under grant 3150-00040B / project no. E4748, and by the German Federal Ministry of Education and Research (BMBF) under grant 01QE2419A. 

\competing 
There were no competing interests to declare which arose during the preparation or publication process of this article.

\data 
The data used in Section~\ref{sec.analysis} can be found in \cite{dataset}.

%
%
%

\bibliographystyle{APT}
\bibliography{bibliography}

\begin{thebibliography}{10}

\bibitem{agresti75}
{\sc Agresti, A.} (1975).
\newblock On the extinction times of varying and random environment branching
  processes.
\newblock {\em Journal of Applied Probability\/} {\bf 12,} 39--46.

\bibitem{aksoy2019}
{\sc Aksoy, S.~G., Purvine, E., Cotilla-Sanchez, E. and Halappanavar, M.}
  (2019).
\newblock A generative graph model for electrical infrastructure networks.
\newblock {\em Journal of Complex Networks\/} {\bf 7,} 128--162.

\bibitem{bertacchi16}
{\sc Bertacchi, D., Rodriguez, P.~M. and Zucca, F.} (2020).
\newblock Galton–{W}atson processes in varying environment and accessibility
  percolation.
\newblock {\em Brazilian Journal of Probability and Statistics\/} {\bf 34,} pp.
  613--628.

\bibitem{broman08}
{\sc Broman, E. and Meester, R.} (2008).
\newblock Survival of inhomogeneous {G}alton-{W}atson processes.
\newblock {\em Advances in Applied Probability\/} {\bf 40,} 798--814.

\bibitem{fearn72}
{\sc Fearn, D.~H.} (1972).
\newblock Galton-{W}atson processes with generation dependence.
\newblock {\em Proceedings of the Sixth Berkley Symposium on Mathematical and
  Statistical Probabilitity\/} {\bf 4,} 159--172.

\bibitem{fearn81}
{\sc Fearn, D.~H.} (1981).
\newblock A fixed-point property for {G}alton-{W}atson processes with
  generation dependence.
\newblock {\em Journal of Applied Probability\/} {\bf 18,} 514--519.

\bibitem{Foster76}
{\sc Foster, J.~H. and Goettge, R.~T.} (1976).
\newblock The rates of growth of the {G}alton-{W}atson process in varying
  environment.
\newblock {\em Journal of Applied Probability\/} {\bf 13,} 144--147.

\bibitem{gao15}
{\sc Gao, Z. and Zhang, Y.} (2015).
\newblock Limit theorems for a {G}alton-{W}atson process with immigration in
  varying environments.
\newblock {\em Bulletin of the Malaysian Mathematical Sciences Society\/} {\bf
  38,} 1551--1573.

\bibitem{giacomarra24}
{\sc Giacomarra, F., Bet, G. and Zocca, A.} (2024).
\newblock Generating synthetic power grids using exponential random graph
  models.
\newblock {\em PRX Energy\/} {\bf 3,} 023005.

\bibitem{Hassani04052025}
{\sc Hassani, S., Olsen, R.~L., Schwefel, H.-P. and and, J. D.~B.} (2025).
\newblock Invariance-based control for cyber-physical electrical distribution
  grids.
\newblock {\em International Journal of Control\/} {\bf 98,} 1100--1110.

\bibitem{KRISTENSEN2018122}
{\sc le~Fevre~Kristensen, T., {Løvenstein Olsen}, R., Pedersen, R., Iov, F.
  and Schwefel, H.-P.} (2018).
\newblock Active power reference tracking in electricity distribution grids
  over non-ideal communication networks.
\newblock {\em International Journal of Electrical Power \& Energy Systems\/}
  {\bf 102,} 122--130.

\bibitem{ma17}
{\sc Ma, S., Yu, Y. and Zhao, L.} (2017).
\newblock Dual-stage constructed random graph algorithm to generate random
  graphs featuring the same topological characteristics with power grids.
\newblock {\em Journal of Modern Power Systems and Clean Energy\/} {\bf 5,}
  683--695.

\bibitem{dataset}
{\sc Olsen, R.~L.}
\newblock Data for modeling electrical distribution networks with inhomogeneous
  {G}alton-{W}atson trees.
\newblock Available at
  \url{https://vbn.aau.dk/da/datasets/data-til-modeling-electrical-distribution-networks-with-inhomogen}.
\newblock Accessed: 2025-04-30.

\bibitem{olsen24}
{\sc Olsen, R.~L., Hassani, S., Pedersen, T., Rasmussen, J.~G. and Schwefel,
  H.-P.} (2024).
\newblock A stochastic approach to estimate distribution grid state with
  confidence regions.
\newblock {\em (Submitted for publication, available on arXiv.)\/}.

\bibitem{schultz2014random}
{\sc Schultz, P., Heitzig, J. and Kurths, J.} (2014).
\newblock A random growth model for power grids and other spatially embedded
  infrastructure networks.
\newblock {\em The European Physical Journal Special Topics\/} {\bf 223,}
  1189--1200.

\bibitem{shahraeini24}
{\sc Shahraeini, M.} (2024).
\newblock Modified {E}rdős–{R}ényi random graph model for generating
  synthetic power grids.
\newblock {\em IEEE Systems Journal\/} {\bf 18,} 96--107.

\end{thebibliography}

\end{document}